\documentclass[12pt,reqno]{amsart}
\usepackage{amsmath}
\usepackage{latexsym}
\usepackage{amsfonts}
\usepackage{amssymb}
\usepackage{color}
\usepackage{bbm,dsfont}
\usepackage{graphicx}

\newtheorem{proposition}{Proposition}

\newtheorem{lemma}{Lemma}

\theoremstyle{definition}

\newtheorem{example}{Example}
\newtheorem{definition}{Definition}

\newcommand{\real}{\mathbb R} 
\newcommand{\complex}{\mathbb C} 
\newcommand{\half}{\tfrac{1}{2}} 
\newcommand{\mo}[1]{\left| #1 \right|} 

\newcommand{\hi}{\mathcal{H}} 
\newcommand{\lh}{\mathcal{L(H)}} 
\newcommand{\sh}{\mathcal{S(H)}} 
\newcommand{\eh}{\mathcal{E(H)}} 
\newcommand{\ip}[2]{\left\langle\,#1\,|\,#2\,\right\rangle} 
\newcommand{\kb}[2]{|#1\,\rangle\langle\,#2|} 
\newcommand{\no}[1]{\left\|#1\right\|} 
\newcommand{\tr}[1]{\textrm{tr}\left[#1\right]} 
\newcommand{\lb}[1]{lb(#1)} 
\newcommand{\id}{I} 
\newcommand{\nul}{O} 

\newcommand{\salg}{\mathcal{F}} 
\newcommand{\var}{\textrm{Var}} 
\newcommand{\bor}[1]{\mathcal{B}(#1)} 
\newcommand{\ltwo}[1]{L^2(#1)} 

\newcommand{\va}{\mathbf{a}} 
\newcommand{\vb}{\mathbf{b}} 
\newcommand{\vc}{\mathbf{c}} 
\newcommand{\vg}{\mathbf{g}} 
\newcommand{\vn}{\mathbf{n}} 
\newcommand{\vnn}{\hat\vn} 
\newcommand{\vsigma}{\boldsymbol{\sigma}} 

\newcommand{\Aaa}{A(\alpha,\va)} 

\newcommand{\A}{\mathsf{A}}
\newcommand{\B}{\mathsf{B}}
\newcommand{\C}{\mathsf{C}}
\newcommand{\E}{\mathsf{E}}
\newcommand{\F}{\mathsf{F}}
\newcommand{\G}{\mathsf{G}}
\renewcommand{\P}{\mathcal{P}}
\newcommand{\Ea}{\mathsf{E}^{1,\va}} 
\newcommand{\Eb}{\mathsf{E}^{1,\vb}} 
\newcommand{\Ec}{\mathsf{E}^{1,\vc}} 
\newcommand{\Eaa}{\mathsf{E}^{\alpha,\mathbf{a}}} 
\newcommand{\Ebb}{\mathsf{E}^{\beta,\mathbf{b}}} 
\newcommand{\Q}{\mathsf{Q}}
\renewcommand{\P}{\mathsf{P}}
\newcommand{\Qmu}{\mathsf{Q}^{\mu}}
\newcommand{\Pnu}{\mathsf{P}^{\nu}}

\begin{document}

\title{Notes on Joint Measurability of Quantum Observables}

\begin{abstract}
For sharp quantum observables the following facts hold:
(i) if we have a collection of sharp observables and each pair of them is jointly measurable, then they are jointly measurable all together; (ii) if two sharp observables are jointly measurable, then their joint observable is unique and it gives the greatest lower bound for the effects corresponding to the observables; (iii) if we have two sharp observables and their every possible two outcome partitionings are jointly measurable, then the observables themselves are jointly measurable. We show that, in general, these properties do not hold. Also some possible candidates which would accompany joint measurability and generalize these apparently useful properties are discussed. 
\end{abstract}

\author[Heinosaari]{Teiko Heinosaari}
\address{Teiko Heinosaari, Department of Physics, University of Turku, 20014 Turku, Finland}
\email{heinosaari@gmail.com}

\author[Reitzner]{Daniel Reitzner}
\address{Daniel Reitzner, Research Center for Quantum Information, Slovak Academy of Sciences, D\'ubravsk\'a cesta 9, 845 11 Bratislava, Slovakia}
\email{daniel.reitzner@savba.sk}

\author[Stano]{Peter Stano}
\address{Peter Stano, Research Center for Quantum Information, Slovak Academy of Sciences, D\'ubravsk\'a cesta 9, 845 11 Bratislava, Slovakia}
\email{peter.stano@savba.sk}

\maketitle

\section{Introduction}\label{sec:intro}

The fact that not all pairs of quantum observables are jointly measurable is usually mentioned as one of the characteristic features of quantum mechanics. Joint measurability means that measurements of two given observables can be interpreted as originating in a measurement of a single (joint) observable. Joint (non-)measurability  is hence related to the limitations of measuring and manipulating quantum objects.

Traditionally, joint measurability was taken to be synonym for commutativity. This is indeed the case if observables are represented solely by selfadjoint operators. However, nowadays it is recognized that the correct mathematical representation for quantum observables is given by positive operator valued measures \cite{PSAQT82}, \cite{OQP97}, and selfadjoint operators correspond only to a specific class of ideal observables (here called sharp observables). This generalization leads to the conclusion that commutativity is only a limited criterion for joint measurability in this more general setting \cite{LaPu97}, \cite{Lahti03}. Namely, commutative observables are jointly measurable but the converse need not be true.

There are numerous studies on joint measurements of some special collections of observables, the most prominent cases being the position--momentum pair and spin observables along different axes. In these notes we are interested in the general features of joint measurability rather than any special case. However, the conclusions we obtain are based on recent progress in the above mentioned examples.   

Unlike in the case of sharp observables (projection valued measures or equivalently selfadjoint operators), there is no general theory for joint measurability of general observables (positive operator valued measures). Two sharp observables are jointly measurable if and only if they commute, but this kind of a simple criterion is not known for general observables.

For sharp observables the following facts hold:
(i) if we have a collection of sharp observables and each pair of them is jointly measurable, then they are jointly measurable all together; (ii) if two sharp observables are jointly measurable, then their joint observable is unique and it gives the greatest lower bound for the effects corresponding to the observables; (iii) if we have two sharp observables and their every possible two outcome partitionings are jointly measurable, then the observables themselves are jointly measurable. 

These properties are often useful. Namely, property (i) implies that to test the joint measurability of a collection of observables, it would be enough to test them pairwisely. Property (ii) simplifies the problem by fixing the form of a joint observable, while  
property (iii) implies that joint measurability of two observables reduces to joint measurability of certain two outcome observables. 

In Sections 2--3 we first define the needed notions and then provide some examples. We then study the properties (i)--(iii) (in Sections 4--6, respectively) and we show that, in general, they do not hold. We, however, also suggest some possible candidates which would accompany joint measurability and replace the apparently useful properties (i)-(iii). We note that this work is not intended to give the final answers in these issues but to provoke further investigations.

\section{Joint measurability of quantum observables}

Let $\hi$ be a complex separable Hilbert space and $\lh$ the set of bounded
operators on $\hi$. A positive operator $\varrho\in\lh$ having trace one is
called a \emph{state} and we denote by $\sh$ the set of all states. A
positive operator $A$ bounded from above by the unity operator $\id$ is called
an \emph{effect} and the set of all effects is denoted by $\eh$. We say
that the null operator $\nul$ and the unity operator $\id$ are \emph{trivial
effects}. 

\begin{definition}
\label{def:observable}
Let $\Omega$ be a nonempty set and $\salg$ a $\sigma$-algebra\footnote{For convenience, we always assume that all the one element sets $\{x\}$ belong to the $\sigma$-algebra $\salg$.} of subsets of $\Omega$. A mapping $\A:\salg\to\eh$ is  an \emph{observable} if
the set function 
\begin{equation*}
\salg\ni X\mapsto \ip{\psi}{\A(X)\psi}\in\real
\end{equation*}
is a probability measure for all unit vectors $\psi\in\hi$. The measurable space $(\Omega,\salg)$ is called an \emph{outcome space} of $\A$.
\end{definition}

For an observable $\A:\salg\to\eh$ and a state $\varrho\in\sh$, we denote by $p^{\A}_{\varrho}$ the following mapping from $\salg$ to the interval $[0,1]$,
\begin{equation*}
p^{\A}_{\varrho}(X)=\tr{\varrho\A(X)} \, .
\end{equation*}
The properties of $\A$, resulting from Definition \ref{def:observable}, guarantee that $p^{\A}_{\varrho}$ is a probability measure. The number $p^{\A}_{\varrho}(X)$ is interpreted as the probability of getting measurement outcome $x$ belonging to $X$, when the system is in the state $\varrho$ and the observable $\A$ is measured.

As a short-hand notation, we denote $\A(x)\equiv \A(\{x\})$ for $x\in\Omega$. If the set $\Omega$ is countable, then the collection of effects $\A(x)$, $x\in\Omega$, determines the observable $\A$. Namely, for any $X\subseteq\Omega$, we have
\begin{equation*}
\A(X)=\sum_{x\in X} \A(x) \, .
\end{equation*}
In particular, the normalization condition $\A(\Omega)=\id$ reads $\sum_{x\in\Omega} \A(x)=\id$.

For $\sigma$-algebras $\salg_1,\salg_2,\ldots,\salg_n$, we denote by $\salg_1\otimes\salg_2\otimes\cdots\otimes\salg_n$ the product $\sigma$-algebra, generated by the sets of the form $X_1\times X_2\times\cdots\times X_n$, $X_j\in\salg_j$.

\begin{definition}
Observables $\A^1,\A^2,\ldots,\A^n$ are \emph{jointly measurable} if there exists an observable $\G$, defined on the outcome space $(\Omega_1\times\Omega_2\times\cdots\times\Omega_n,\salg_1\otimes\salg_2\otimes\cdots\otimes\salg_n)$, such that
\begin{eqnarray*}
&& \G(X_1\times\Omega_2\times\cdots\times\Omega_n) = \A^1(X_1) \\
&& \G(\Omega_1\times X_2\times\cdots\times\Omega_n) = \A^2(X_2) \\
&& \quad \vdots \\
&& \G(\Omega_1\times\Omega_2\times\cdots\times X_n) = \A^n(X_n)
\end{eqnarray*}
for all $X_1\in\salg_1,\ldots, X_n\in\salg_n$. In this case $\G$ is a \emph{joint observable} of $\A^1,\A^2,\ldots,\A^n$.
\end{definition}

To illustrate the definition of joint measurability, let $\A$ and $\B$ be two observables with finite outcome spaces $\Omega_{\A}=\{x_1,x_2,\ldots,x_n\}$ and $\Omega_{\B}=\{y_1,y_2,\ldots,y_m\}$, respectively. The condition for joint measurability of $\A$ and $\B$ is the existence of such $\G$ defined on the outcome space $\Omega_{\A}\times\Omega_{\B}$, for which
\begin{equation}\label{eq:sum_rule}
\A(x_i)=\sum_{j=1}^m \G(x_i,y_j), \quad \B(y_j)=\sum_{i=1}^n \G(x_i,y_j) \, .
\end{equation}

If observable $\A$ has projections as its values (i.e. $\A(X)^2=\A(X)$ for all $X\in\salg$), it is called \emph{sharp observable}. Two sharp observables are jointly measurable if and only if they commute. (For convenience, this well-known result is proved in a slightly more general form in Appendix as we use it constantly.) Commutativity of $\A$ and $\B$ means that all pairs of their effects commute, i.e., $[\A(X),\B(Y)]=0$ for all $X$, $Y$. We stress that for general observables commutativity and joint measurability are not equivalent concepts.

\section{Examples}

\subsection{Simple qubit observables}\label{sec:example}

A two outcome observable is called \emph{simple}. Qubit observables are those defined on the qubit Hilbert space $\complex^2$. One possible characterization of a simple qubit observable is by providing an effect corresponding to one measurement outcome, since the other one is fixed by the normalization. A selfadjoint operator on $\complex^2$ can be parametrized by four real parameters $(\alpha,\va)\in\real^4$,
\begin{equation}
\Aaa:=\half \left(\alpha \id + \va\cdot\vsigma \right) \, ,
\label{qubit effect definition}
\end{equation}
where $\vsigma\equiv(\sigma_1,\sigma_2,\sigma_3)$ is the triplet of Pauli matrices. Such operator $\Aaa$ is an effect if and only if
\begin{equation}
\no{\va} \leq \alpha \leq 2-\no{\va} \, ,
\end{equation}
and a nontrivial projection if and only if
\begin{equation}
\alpha=\no{\va}=1 \, .
\end{equation}
We denote by $\Eaa$ the simple qubit observable corresponding to an effect $\Aaa$. The labeling of the measurement outcomes of $\Eaa$ is irrelevant for joint measurability questions, but for convenience we fix the outcomes to be $0$ and $1$, so that $\Eaa$ is defined as
\begin{equation*}
\Eaa(1)=\Aaa\, , \quad \Eaa(0)=\id-\Aaa \, .
\end{equation*}

In Proposition \ref{prop:two-qubit} we summarize some known results on joint measurability of two simple qubit observables which we will use later. 

\begin{proposition}\label{prop:two-qubit}
Let $\Eaa$ and $\Ebb$ be two qubit observables.
\begin{itemize}
\item (Busch \cite[Theorem 4.5.]{Busch86}) If $\alpha=\beta=1$, then $\Eaa$ and $\Ebb$ are jointly measurable if and only if
\begin{equation}\label{jm:busch}
\no{\va + \vb} + \no{\va - \vb} \leq 2 \, .
\end{equation}
\item (Moln\'ar \cite[Lemma 2]{Molnar01b}) If $\alpha=\no{\va}$, $\beta=\no{\vb}$ and $\va\nparallel\vb$, then $\Eaa$ and $\Ebb$ are jointly measurable if and only if
\begin{equation}\label{jm:molnar}
\no{\va+\vb} + \no{\va} +\no{\vb}\leq 2 \, .
\end{equation}
\item (Liu et al. \cite[Theorem 1]{LiLiYuCh07}) If $\alpha=1$ and $\va\bot\vb$, then $\Eaa$ and $\Ebb$ are jointly measurable if and only if
\begin{equation}\label{jm:chinese}
2 \no{\va} \leq \sqrt{\beta^2 - \no{\vb}^2} + \sqrt{(2-\beta)^2 - \no{\vb}^2} \, .
\end{equation}
\end{itemize}
\end{proposition}
All these cases follow from a general theorem, proved recently in different forms in \cite{BuSc08, StReHe08, YuLiLiOh08}. 

A joint observable $\G$ for observables $\Eaa$ and $\Ebb$, if it exists, is a four outcome observable, and Eq.~\eqref{eq:sum_rule} in this case reads
\begin{equation*}\begin{split}
\Eaa(1)& =\G(1,1)+\G(1,0)\, , \qquad \Ebb(1)=\G(1,1)+\G(0,1)\, , \\
\Eaa(0)&=\G(0,0)+\G(0,1)\, , \qquad \Ebb(0)=\G(0,0)+\G(1,0)\, .
\end{split}\end{equation*}
From these equations and the normalization condition follows that $\G$ is determined by just one effect, say $\G(1,1)$. All other effects forming $\G$ are determined by this effect and the observables $\Eaa$ and $\Ebb$. Effect $\G(1,1)$ is of the form given by Eq.~\eqref{qubit effect definition}, hence $\G(1,1)=\half (\gamma\id+\vg\cdot\vsigma)$ for some parameters $(\gamma,\vg)\in\real^4$. The requirement that the other three operators $\G(1,0)$, $\G(0,1)$ and $\G(0,0)$ are actually effects leads then to constrains for the parameters $\gamma$ and $\vg$. Inspection of these constrains gives the results cited in Proposition \ref{prop:two-qubit}.

In Proposition \ref{prop:three-qubit} we combine the results of Busch \cite{Busch86} and Andersson et al. \cite{BrAn07} concerning the joint measurability of three simple qubit observables. 

\begin{proposition}\label{prop:three-qubit}
Let $\va,\vb,\vc$ be three orthogonal vectors and $\Ea, \Eb, \Ec$ the corresponding qubit observables. Then $\Ea$, $\Eb$ and $\Ec$ are jointly measurable if and only if
\begin{equation}\label{jm:three}
\no{\va}^2+\no{\vb}^2+\no{\vc}^2 \leq 1 \, . 
\end{equation}
\end{proposition}

\subsection{Position and momentum observables}

Let us consider the usual position and momentum observables $\Q$ and $\P$ of a free spin-0 particle moving in one dimension. The outcome space for both of these observables is the Borel $\sigma$-algebra $\bor{\real}$. The Hilbert space for this system is  $\ltwo{\real}$, the space of square integrable functions. For a pure state $\varrho=\kb{\psi}{\psi}$, the corresponding probability measures are
\begin{equation*}
p^{\Q}_{\varrho}(X)=\int_X \mo{\psi(q)}^2 dq\, ,\quad p^{\P}_{\varrho}(Y)=\int_Y \mo{\hat{\psi}(p)}^2 dp \, ,
\end{equation*}
where $\hat{\psi}$ is the Fourier transform of $\psi$. 

The observables $\Q$ and $\P$ are sharp and they do not commute. Therefore, there is no joint observable for $\Q$ and $\P$. It was pointed out by Davies \cite{QTOS76} and Ali \& Prukovecki \cite{AlPr77a} that there exist observables $\Qmu$ and $\Pnu$ which can be considered as approximate versions of $\Q$ and $\P$, respectively, and which have a joint observable. Here we shortly recall this example; for more details, we refer to \cite{BuHeLa07} and references given therein. 

For a probability measure $\mu$ on $\real$, we define observable $\Qmu$ by formula
\begin{equation}
\Qmu(X)=\int \Q(X-q) \ d\mu(q)\, , \quad X\in\bor{\real}\, .
\end{equation}
The observable $\Pnu$ is defined in a similar way using the canonical momentum observable $\P$ and a probability measure $\nu$. 

If the probability measures $\mu$ and $\nu$ are properly chosen, then $\Qmu$ and $\Pnu$ have a joint observable. The characterization of jointly measurable pairs of $\Qmu$ and $\Pnu$ was completed by Carmeli et al. \cite{CaHeTo05} and we summarize it in the following proposition.

\begin{proposition}\label{prop:pos-mom-joint-condition}
Let $\Qmu$ and $\Pnu$ be position and momentum observables, respectively. They are jointly measurable if and only if there exists a positive trace class operator $T\in\lh$ having trace 1 and satisfying
\begin{equation}\label{eq:pos-mom-joint-condition}
\mu=p^{\Q}_{\Pi T \Pi}\, ,\qquad \nu=p^{\P}_{\Pi T\Pi} \, .
\end{equation}
Here $\Pi:\hi\to\hi$ is the parity operator $\Pi\psi(x)=\psi(-x)$. 
\end{proposition}

For instance, if the operator $T$ in Proposition \ref{prop:pos-mom-joint-condition} is a one-dimensional projection $T=\kb{\phi}{\phi}$, then the condition \eqref{eq:pos-mom-joint-condition} reads
\begin{equation}\label{eq:pos-mom-joint-condition-2}
d \mu(q)=\mo{\phi(-q)}^2 dq\, ,\qquad d \nu(p)=\mo{\hat{\phi}(-p)}^2 dp \, .
\end{equation}

Let $\Qmu$ and $\Pnu$ be a jointly measurable pair of a position and momentum observable and let $T$ be an operator satisfying conditions of Proposition \ref{prop:pos-mom-joint-condition}. A joint observable $\G_T$ for $\Qmu$ and $\Pnu$ is then given by formula
\begin{equation}\label{eq:GT}
\G_T(Z)=\frac{1}{2\pi} \int_Z U_{qp}TU_{qp}^\ast \ dqdp\, ,\quad Z\in\bor{\real^2}\, ,
\end{equation}
where $U_{qp}$ is the unitary operator defined by $U_{qp}\psi(x)=e^{ipx}\psi(x-q)$.

\section{Pairwise joint measurability}

In this section we study a question whether for a set of observables, the existence of a joint observable of the whole set is equivalent to the existence of joint observables for all pairs of observables from the set --- a fact that holds for sharp observables.

\begin{proposition}\label{prop:2sharp->3joint}
If observables $\A,\B$, and $\C$ are pairwisely jointly measurable and (at least) two of them are sharp, then the triplet $\A,\B,\C$ is jointly measurable.
\end{proposition}

\begin{proof}
Let, for instance, $\A$ and $\B$ be sharp observables. Since $\B$ and $\C$ are jointly measurable, they have a joint observable $\G$ which is determined by condition $\G(Y\times Z)=\B(Y)\C(Z)$ for every $Y\in\salg_\B,Z\in\salg_\C$ (see Proposition \ref{prop:commute} in Appendix). Fix $X\in\salg_\A$. Since $\A$ is jointly measurable and commutes with both $\B$ and $\C$,  the projection $\A(X)$ commutes with $\B(Y)\C(Z)$ for any $Y\in\salg_\B,Z\in\salg_\C$. Thus, the set functions $\A(X)\G(\cdot)$ and $\G(\cdot)\A(X)$ agree on the product sets of the form $Y\times Z$, which implies that they are the same (see Lemma \ref{lemma:GXY} in Appendix). We conclude that $\A$ and $\G$ commute and hence, are jointly measurable. This implies that the triplet $\A,\B,\C$ is jointly measurable.
\end{proof}

However, generally, if observables $\A,\B$ and $\C$ are pairwisely jointly measurable, it does not imply that the triplet $\A,\B,\C$ is jointly measurable. This is demonstrated in the following example.

\begin{example}
Let $\va,\vb,\vc\in\real^3$ be three orthogonal vectors, each of length $l$. The joint measurability condition \eqref{jm:busch} for any pair of  $\Ea$, $\Eb$, $\Ec$ then gives $l\leq\frac{1}{\sqrt{2}}$, while the joint measurability condition \eqref{jm:three} of all three means that $l\leq\frac{1}{\sqrt{3}}$. Therefore, the choice $\frac{1}{\sqrt{3}}<l\leq\frac{1}{\sqrt{2}}$ leads to observables $\Ea,\Eb,\Ec$ which are pairwisely jointly measurable but not  jointly measurable all together. 
\end{example}

\begin{figure}
\centerline{\includegraphics[scale=1]{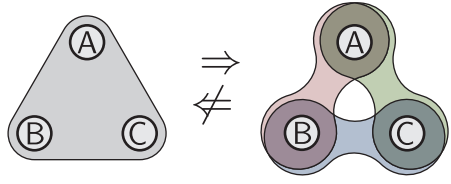}}
\caption{\label{fig:3vs2}
Joint observable for three observables $\A$, $\B$, $\C$ implies that there exist joint observables for each possible pair, $\A$-$\B$, $\A$-$\C$, and  $\B$-$\C$, but not vice versa.}
\end{figure}

We leave it as an open problem whether the conditions in Proposition \ref{prop:2sharp->3joint} can be relaxed. For instance, if it is sufficient that just one observable in the triplet is sharp. Summarizing this section, we demonstrated that pairwise joint measurability does not guarantee the existence of a joint observable for the whole set of observables, depicted in Fig. \ref{fig:3vs2}.

\section{Uniqueness of joint observables}\label{sec:uniqueness}

It is a well known fact that a joint observable is, in general, not unique. This is not surprising as in the definition of the joint observable requirements are set for the margins only and not for all effects in the range of $\G$.

In this section we study some conditions which would guarantee that a joint observable for a given pair of observables is unique. We start by illustrating this problem in the case of position and momentum observables.

\begin{example}\label{ex:gaussian}
Let $\Qmu$ and $\Pnu$ be jointly measurable position and momentum observables.
Assume that the probability measures $\mu$ and $\nu$ are such that the variances (i.e. squares of the standard deviations) $\var(\mu)$ and $\var(\nu)$ satisfy the minimum variance equation $\var(\mu)\cdot\var(\nu)=\frac{1}{4}$. In this case, an operator $T$ satisfying \eqref{eq:pos-mom-joint-condition} is a one dimensional projection $T=\kb{\phi}{\phi}$, where $\phi$ is a Gaussian function. This is consequence of the well-known fact that the functions satisfying the minimum variance equation are Gaussians. A Gaussian function $\phi$ is, on the other hand, determined up to a phase factor by functions $\mo{\phi(\cdot)}$ and $\mo{\hat{\phi}(\cdot)}$ and these are determined in \eqref{eq:pos-mom-joint-condition-2}. Therefore, there is a unique operator $T$ such that $\G_T$, defined in \eqref{eq:GT}, is a joint observable for $\Qmu$ and $\Pnu$. 

Generally, however, there may be two operators $T\neq T'$ such that both $\G_{T}$ and $\G_{T'}$ are joint observables of $\Qmu$ and $\Pnu$ (see e.g. Example 1 in \cite{CaHeTo05}). Moreover, up to authors' knowledge it is not known whether all joint observables of a jointly measurable pair $\Qmu$ and $\Pnu$ are of the form $\G_T$ for some trace class operator $T$. This example shows that there may exist several joint observables, that is, a joint observable is not always unique.
\end{example}

From the previous discussion naturally arises the question in which situations the joint observable is unique. On the other hand, if there are many joint observables, it would be convenient to have a way to compare different joint observables and perhaps to pick out the "best" one.

To proceed in our investigation, let us recall that the set of effects $\eh$ is a partially ordered set. For effects $A,B\in\eh$, the ordering $A\leq B$ means that
\begin{equation*}
\ip{\psi}{A\psi}\leq \ip{\psi}{B\psi}
\end{equation*}
for every $\psi\in\hi$. This is equivalent to the condition that
\begin{equation*}
\tr{\varrho A}\leq \tr{\varrho B}
\end{equation*}
for every $\varrho\in\sh$. Hence, the effect $B$ gives greater or equal probability in every state than $A$.

For two effects $A$ and $B$, we denote by $\lb{A,B}$ the set of their lower bounds, that is,
$$
\lb{A,B}:=\{ C\in\eh \mid C\leq A,C\leq B \}\, .
$$
This is always a non-empty set as $\nul\in\lb{A,B}$.

Let us have two jointly measurable observables $\A$ and $\B$ and let $\G$ be their joint observable. Then for every $X\in\salg_\A,Y\in\salg_\B$, the effect $\G(X\times Y)$ is in $\lb{\A(X),\B(Y)}$ since
\begin{equation*}
\A(X)=\G(X\times Y)+\G(X\times \neg Y)\geq \G(X\times Y) \, ,
\end{equation*}
and 
\begin{equation*}
\B(Y)=\G(X\times Y)+\G(\neg X\times Y)\geq \G(X\times Y) \, .
\end{equation*}

If either $\A$ or $\B$ (or both) is sharp, then the effect $\G(X,Y)$ is a special element in set $\lb{\A(X),\B(Y)}$ --- it is the greatest one. Indeed, it is a direct consequence of \cite[Corollary 2.3.]{MoGu99} that for every $X\in\salg_\A,Y\in\salg_\B$, the effect $\G(X\times Y)$ is the greatest element of $\lb{\A(X),\B(Y)}$, i.e., for every $C\in\lb{\A(X),\B(Y)}$, the relation $C\leq\G(X\times Y)$ holds. Moreover, in this case $\A$ and $\B$ commute and their joint observable $\G$ is unique, as shown in Proposition \ref{prop:commute} in Appendix.

These observations motivate the following definition.
 
\begin{definition}
A joint observable $\G$ of observables $\A$ and $\B$ is \emph{the greatest joint observable} if for every $X\in\salg_{\A},Y\in\salg_{\B}$, the effect $\G(X\times Y)$ is the greatest element of $\lb{\A(X),\B(Y)}$.
\end{definition}

The above property of being greatest and the uniqueness of joint observable are related in the following way.

\begin{proposition}\label{prop:infimum}
Let $\A$ and $\B$ be two jointly measurable observables and let $\G$ be their joint observable. If $\G$ is the greatest joint observable, then it is the unique joint observable.
\end{proposition}

\begin{proof}
Let $\F$ and $\G$ be joint observables of $\A$ and $\B$ with $\G$ being the greatest one. We make a counter assumption that $\F\neq\G$. Then $\F(X\times Y)\neq\G(X\times Y)$ for some $X\in\salg_\A,Y\in\salg_\B$ by Lemma \ref{lemma:GXY}. Since $\G(X\times Y)$ is the greatest element in $\lb{\A(X),\B(Y)}$, this means that $\F(X\times Y)<\G(X\times Y)$ for any choice of $X$, $Y$. But this implies that
\begin{eqnarray*}
\id &=& \F(\Omega_\A\times\Omega_\B) \\
&=& \F(X\times Y)+\F(X\times \neg Y)+\F(\neg X\times Y)+\F(\neg X\times \neg Y) \\ 
&<& \G(X\times Y)+\G(X\times \neg Y)+\G(\neg X\times Y)+\G(\neg X\times \neg Y) \\
&=& \G(\Omega_\A\times\Omega_\B)=\id\, ,
\end{eqnarray*}
which is a contradiction. Thus, $\F=\G$.
\end{proof}

We conclude from Proposition \ref{prop:infimum} that two observables $\A$ and $\B$ can have at most one greatest joint observable. In the following example we demonstrate that a joint observable may be unique even if it is not greatest. 

\begin{example}\label{ex:unique}
Let $\va,\vb\in\real^3$ be two vectors satisfying
\begin{equation}
\label{eq:boundary}
\no{\va+\vb}+\no{\va-\vb}=2\, .
\end{equation}
According to Proposition \ref{prop:two-qubit} the corresponding qubit observables $\Ea$ and $\Eb$ are jointly measurable. As explained in \cite[Appendix 2]{StReHe08}, their joint observable $\G$ is unique and given by
\begin{equation}
\label{eq:joint_BuHe}
\G(i,j)=\no{\vn_{ij}}\half \left( \id + \vnn_{ij}\cdot\vsigma\right), \quad i,j\in\{0,1\} \, ,
\end{equation}
where $\vn_{ij}=\half[(-1)^{i+1}\va+(-1)^{j+1}\vb]$ and $\vnn_{ij}=\vn_{ij}/\no{\vn_{ij}}$. 

It is easy to see\footnote{Alternatively, one can apply \cite[Theorem 2]{GuGr96} to see that the greatest element of $lb(\Ea(i), \Eb(j))$ does not exist.} that $\G(i,j)$ is not the greatest element in the set $\lb{\Ea(i), \Eb(j)}$. Fix $0<t<\half$ and choose $\gamma$ from the interval $(t,\half)$. Then $C=\half (\gamma\id + t(\va+\vb)\cdot\vsigma )$ is an effect and $C\leq \Ea(1)$, $C\leq \Eb(1)$, while it is not true that $C\leq\G(1,1)$.  
\end{example}

As the property of being greatest is stronger than the uniqueness property, we are now seeking for a replacement for it.  For two effects $A$ and $B$, the set $\lb{A,B}$ may not have the greatest element. However, there always exists a \emph{maximal} element in $\lb{A,B}$ and every lower bound lies under some maximal lower bound \cite[Theorem 4.5]{MoGu99}. We recall that an effect $C\in\lb{A,B}$ is called maximal in $\lb{A,B}$ if there does not exist $D\in\lb{A,B}$ such that $C< D$. Clearly, if the greatest element in $\lb{A,B}$ exists, then it is the unique maximal element. Generally, however, $\lb{A,B}$ may have many maximal elements.

\begin{definition}
A joint observable $\G$ of observables $\A$ and $\B$ is \emph{maximal joint observable} if for every $X\in\salg_{\A},Y\in\salg_{\B}$, the effect $\G(X\times Y)$ is a maximal element of $\lb{\A(X),\B(Y)}$.
\end{definition}

The obvious questions are now whether maximal joint observables exist and whether this property has some connection to the uniqueness of joint observables. In the following we give partial answers to these questions.

\begin{proposition}\label{prop:maximal}
Let $\A$ and $\B$ be two simple (i.e. two-outcome) observables having a unique joint observable $\G$. Then $\G$ is a maximal joint observable.
\end{proposition}

\begin{proof}
Fix $i,j\in\{0,1\}$. Assume, in contrary, that $\G(i,j)$ is not maximal in $\lb{\A(i),\B(j)}$. This means that there is an  effect $C\in\lb{\A(i),\B(j)}$ satisfying $\G(i,j)<C$. But then
\begin{equation*}
C>\G(i,j)\geq \A(i)+\B(i)-\id,
\end{equation*}
and hence $C$ defines a joint observable $\widetilde{\G}\neq\G$ of $\A$ and $\B$ through formulas
\begin{equation*}
\begin{split}
\widetilde{\G}(i,j)=C,\quad \widetilde{\G}(i,1-j)=\A(i)-C,\quad \widetilde{\G}(1-i,j)=\B(i)-C ,\\
\widetilde{\G}(1-i,1-j)=\id+C-\A(i) - \B(i).
\end{split}
\end{equation*}
This is, however, not possible as $\G$ is assumed to be the unique joint observable.
\end{proof}

As demonstrated in the following example, it may happen that among all joint observables of two jointly measurable observables $\A$ and $\B$, there is no maximal joint observable. 

\begin{example}\label{ex:no-maximal}
Choose a vector $\va\in\real^3$ with $0<\no{\va}<1$. Choose then $\beta\in\real$ such that $\half (1-\no{\va}^2)<\beta < 1-\no{\va}^2$. Finally, choose vector $\vb\in\real^3$ orthogonal to $\va$ and satisfying $\no{\vb}=\beta$. Then corresponding observables $\Ea$ and $\Ebb$ are jointly measurable due to the condition \eqref{jm:chinese}.

As observed in \cite{StReHe08}, the joint observables of $\Ea$ and $\Ebb$ are in one-to-one correspondence with the numbers from the interval $J\equiv[\beta - \half (1-\no{\va}^2), \half (1-\no{\va}^2)]$. Namely, if $\gamma\in J$, then
\begin{equation}
\label{eq:non_unique}
\G(1,1)=\frac{\gamma}{2}(\id+\hat\vb\cdot\vsigma)
\end{equation}
determines a joint observable. (Here $\hat{\vb}$ is the unit vector in the direction of $\vb$.)

Let  $\gamma_1,\gamma_2\in J$ and $\G^1(1,1)$, $\G^2(1,1)$ the corresponding effects, defined in Eq.~(\ref{eq:non_unique}). Then we have $\G^1(1,1) < \G^2(1,1)$ if and only if $\gamma_1<\gamma_2$. On the other hand, the effects $\G^l(0,1)$, $l=1,2$, are given by the expression
\[
\G^{l}(0,1)=\B-\G^{l}(1,1)=\frac{\beta-\gamma_{l}}{2}(\id+\hat\vb\cdot\vsigma).
\]
From this we see that $\G^1(0,1)>\G^2(0,1)$ if and only if $\gamma_1<\gamma_2$ and so the ordering of these effects is opposite to the ordering of the effects $\G^1(1,1)$ and $\G^2(1,1)$. Hence, we conclude that there is no maximal joint observable for $\Ea$ and $\Ebb$.
\end{example}
\renewcommand\arraystretch{2}
\begin{table}
\begin{tabular}{crcl}
\hline
$\A$, $\B$& property of $\G$ & implication & property of $\G$\\
\hline
\hline
sharp & $\G$ exists & $\Rightarrow$ & $\G$ is unique and the greatest\\[1ex]
\hline
general & $\G$ is unique & $\displaystyle{\genfrac{}{}{0pt}{}{\not\Rightarrow}{\Leftarrow}}$ & $\G$ is the greatest \\[1ex]
\hline
simple & $\G$ is unique & $\displaystyle{\genfrac{}{}{0pt}{}{\Rightarrow}{\Leftarrow\,\,\mathrm{(?)}}}$ & $\G$ is maximal\\[1ex]
\hline
general & $\G$ is unique & $\displaystyle{\genfrac{}{}{0pt}{}{\Rightarrow\,\,\mathrm{(?)}}{\Leftarrow\,\,\mathrm{(?)}}}$ & $\G$ is maximal\\[1ex]
\hline\\
\end{tabular}
\caption{\label{tab:uniqueness} Relations between uniqueness of the joint observable $\G$ and its maximality for different classes of observables $\A$ and $\B$. We say that the joint observable is greatest (maximal), if $\G(X\times Y)$ is the greatest (maximal) element in $\lb{\A(X),\B(Y)}$ for all $X,Y$.}
\end{table}
\renewcommand\arraystretch{1}

We summarize the results of this section in Table \ref{tab:uniqueness}. We note that uniqueness of the joint observable is not equivalent to being greatest. We leave it as open problem whether it is equivalent to maximality.

\section{Joint measurability of partitionings}

Let $\A$ be an observable with an outcome space $(\Omega_\A,\salg_\A)$. Instead of measuring $\A$ as such, we can pose a simpler question whether a measurement outcome belongs to a given set $X\in\salg_\A$ or not. Hence, we define a simple observable $\A^X$ with two outcomes '1' and '0',
\begin{equation*}
\A^X(1):=\A(X),\quad \A^X(0):=\A(\neg X)=\id-\A(X)\, .
\end{equation*}
We call $\A^X$ \emph{a partitioning of $\A$ with respect to $X$}.

For instance, if $\Omega_\A$ consists of four elements $x_1,x_2,x_3,x_4$, we can form 16 different partitionings, listed in Table~\ref{tab:partitions}. Hence, partitioning procedure means that several outcomes are identified as one. Two such partitionings are demonstrated in Fig.~\ref{fig:partitions}. 

\begin{figure}
\centerline{\includegraphics[scale=0.7]{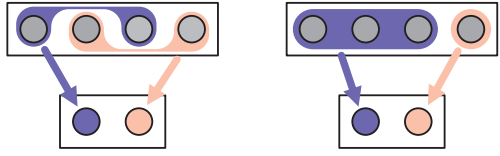}}
\caption{\label{fig:partitions}Examples of two different partitionings.}
\end{figure}

\renewcommand\arraystretch{1.25}
\begin{table}
\begin{tabular}{ccl}
\hline
Type & Number & Choice for $X$ (outcome '1' partitioning)\\
\hline
\hline
0 & 1 & $\emptyset$ (trivial)\\
\hline
1 & 4 & $\{00\}$, $\{01\}$, $\{10\}$, $\{11\}$\\
\hline
2 & 6 & 
\begin{tabular}{@{}l@{}}
$\{00,01\}$, $\{00,10\}$, $\{00,11\}$,\\
$\{01,10\}$, $\{01,11\}$, $\{10,11\}$\end{tabular}\\
\hline
3 & 4 & $\{00,01,10\}$, $\{00,01,11\}$, $\{00,10,11\}$, $\{01,10,11\}$\\
\hline
4 & 1 & $\Omega$ (trivial)\\
\hline\\
\end{tabular}
\caption{\label{tab:partitions}Possible two outcome partitionings of a four element set $\Omega=\{00,01,10,11\}$. The type is determined by the number of outcomes comprised in the '1' partition.}
\end{table}
\renewcommand\arraystretch{1}

Since the partitionings are simple observables, it is often easier to study the joint measurability of such partitionings of $\A$ and $\B$ rather than the joint measurability of $\A$ and $\B$ themselves. These questions are related in the following way.

\begin{proposition}\label{prop:partitionings}
Let $\A$ and $\B$ be two observables. Consider the following conditions: 
\begin{itemize}
\item[(i)] $\A$ and $\B$ are jointly measurable.
\item[(ii)] All partitionings of $\A$ and $\B$ are jointly measurable.
\end{itemize}
Condition $(i)$ implies $(ii)$. If either $\A$ or $\B$ (or both) is sharp, then $(i)$ and $(ii)$ are equivalent.
\end{proposition}

\begin{proof}
Assume that $\A$ and $\B$ are jointly measurable and $\G$ is their joint observable. Let $X\in\salg_\A$ and $Y\in\salg_\B$. Define a four outcome observable $\widetilde{\G}$ as
\begin{eqnarray*}
\widetilde{\G}(1,1) &=& \G(X\times Y)\, , \quad \widetilde{\G}(1,0) = \G(X\times \neg Y) \\
\widetilde{\G}(0,1) &=& \G(\neg X\times Y)\, , \quad \widetilde{\G}(0,0) = \G(\neg X\times \neg Y) \, .
\end{eqnarray*}
Then $\widetilde{\G}$ is a joint observable for $\A^X$ and $\B^Y$, and hence, (i) implies (ii).

Assume then that $\A$ is sharp and (ii) holds. Every partitioning of $\A$ is sharp. Hence, for every $X\in\salg_\A$ and $Y\in\salg_\B$, the observables $\A^X$ and $\B^Y$ commute (see Prop. \ref{prop:commute} in Appendix). But this means that $\A$ and $\B$ commute, which implies that they are jointly measurable.
\end{proof}

Generally, conditions (i) and (ii) are not equivalent if the assumption that either $\A$ or $\B$ is sharp is dropped. This is demonstrated in the following example.

\begin{example}\label{ex:partitionings}
Let $\va,\vb,\vc\in\real^3$ be mutually orthogonal vectors with norm $\frac{1}{\sqrt{2}}$. Observables $\Ea$ and $\Eb$ are jointly measurable and they have a unique joint observable $\G$; see Example \ref{ex:unique}. Similarly, $\Eb$ and $\Ec$ are jointly measurable and we denote by $\F$ their unique joint observable. The observables $\G$ and $\F$ are not jointly measurable. Indeed, their joint measurability would mean that all $\Ea,\Eb,\Ec$ are jointly measurable together, which is not true according to the criterion \eqref{jm:three}. However, every partitioning of $\G$ is jointly measurable with every partitioning of $\F$, as we illustrate in the following.

First of all, for the question of existence of a joint observable the outcome values are irrelevant. It is therefore enough to consider just partitionings of types 1 and 2 from Table \ref{tab:partitions}. 

Consider a case where both $\G$ and $\F$ have partitionings of the type 1. Since $\G(i,j)\leq\half\id$ and $\F(k,l)\leq\half\id$ for every set of indices $i,j,k,l$, these partitionings are jointly measurable in a trivial way; see Example 1 in \cite{BuHe08} . 

Let us then consider partitionings of the type 2. For instance, choose partitioning $X=\{11,10\}$ for $\G$ and $Y=\{11,01\}$ for $\F$. These partitionings then give $\G^X=\Ea$ and $\F^Y=\Ec$, which are jointly measurable. All other partitionings of the type 2 are seen to be jointly measurable in a similar way.

Finally, suppose that $\G$ has partitioning of type 1 and $\F$ of type 2. If the partitioning of $\F$ is chosen to be either $Y=\{11,10\}$ or $Y=\{11,00\}$, then $\F^{Y}=\Eb$ or $\F^{Y}$ is a trivial effect. Hence, both are jointly measurable with any type 1 partitioning of $\G$. If we choose $Y=\{11,01\}$, then $\F^Y=\Ec$. The partitionings of $\G$ are, according to Example \ref{ex:unique}, of the form $\G(i,j)=\E^{\eta,\vn_{ij}}(1)$, with $\eta=\half$, $\vn_{ij}=\half\left[(-1)^{i+1}\va+(-1)^{j+1}\vb\right]$, and $\no{\vn_{ij}}=\half$. The condition \eqref{jm:chinese} holds and therefore these partitionings are again jointly measurable.
\end{example}

It would be interesting to find a sufficient criterion (weaker than sharpness) for $\A$ and $\B$ which would guarantee the equivalence of the conditions (i) and (ii) in Proposition \ref{prop:partitionings}.

\section{Conclusions}

As soon as positive operator valued measures are accepted as the correct mathematical description for quantum observables, the definition of joint measurability follows naturally from the statistical structure of quantum mechanics. Apart from its foundational importance, the existence of a joint observable is of practical use, for example, in secure communication. Commutativity must be seen only as a sufficient criterion for joint measurability. The special form of sharp observables allows them to have only joint observables of definite form.

In this manuscript we have investigated several features which accompany joint measurability in the case of sharp observables. Namely, (i) the equivalence between pairwise and common joint measurability for a set of observables, (ii) uniqueness and maximality of a joint observable, and (iii) joint measurability of observables and their partitionings. Knowledge of such relations would allow us to substantially simplify the search for a joint observable. 

To conclude, we have found that all three features (i)--(iii) are not valid for general observables. Nevertheless, we hope that these observations provide one step for achieving better understanding of the general features of joint observables.

\section*{Appendix}

\begin{definition}
Let $(\Omega,\salg)$ be measurable space. A set function $\A:\salg\to\lh$ is an \emph{operator valued measure} if
\begin{itemize}
\item $\A(\emptyset)=\nul$\, ;    
\item $\ip{\psi}{\A(\cup_{j=1}^{\infty} X_j) \psi} = \sum_{j=1}^{\infty} \ip{\psi}{\A(X_j)\psi}$ for all vectors $\psi\in\hi$ and all sequences $\{X_j\}$ of disjoint sets in $\salg$.
\end{itemize}
\end{definition}

An operator valued measure $\A$ is thus an observable if each operator $\A(X)$ is positive and $\A(\Omega)=\id$.

\begin{lemma}\label{lemma:GXY}
Let $\G$ and $\F$ be two operator valued measures defined on $\Omega_1\times\Omega_2$. If 
\begin{equation}\label{eq:GXY=FXY}
\G(X\times Y)=\F(X\times Y)\quad \forall X\in\salg_1, Y\in\salg_2\, ,
\end{equation}
then $\G=\F$.
\end{lemma}

\begin{proof}
Fix $\psi\in\hi$. Condition (\ref{eq:GXY=FXY}) implies that the complex measures $\ip{\psi}{\G(\cdot)\psi}$ and $\ip{\psi}{\F(\cdot)\psi}$ agree on all rectangles and hence also on all countable disjoint unions of rectangles. Therefore, $\ip{\psi}{\G(Z)\psi}=\ip{\psi}{\F(Z)\psi}$ for every $Z\in\salg_1\otimes\salg_2$. Since this holds for any $\psi\in\hi$, we conclude
 that $\G(Z)=\F(Z)$ for every $Z\in\salg_1\otimes\salg_2$. Thus, $\G=\F$. 
\end{proof}

\begin{proposition}\label{prop:commute}
Let $\A$ and $\B$ be jointly measurable observables. If (at least) one of them is sharp, then they commute and they have a unique joint observable $\G$. This joint observable is determined by the condition
\begin{equation}\label{eq:GAB}
\G(X\times Y)=\A(X)\B(Y) \qquad \forall X\in\salg_\A,Y\in\salg_\B\,  .
\end{equation}
\end{proposition}

\begin{proof}
Let, for instance, $\A$ be a sharp observable and suppose that $\G$ is a joint observable for $\A$ and $\B$. Since 
\begin{equation*}
\G(X\times Y)\leq \G(X\times \Omega_\B)=\A(X)\, ,
\end{equation*}
it follows that the range of $\G(X\times Y)$ is contained in the range of $\A(X)$. Therefore,
\begin{equation*}
\A(X)\G(X\times Y)=\G(X\times Y) \, ,
\end{equation*}
and taking adjoint on both sides we also get
\begin{equation*}
\G(X\times Y)\A(X)=\G(X\times Y)\, .
\end{equation*}
Applying these results to the complement set $\neg X$ we get
\begin{equation*}
\A(X)\G(\neg X\times Y)=\left( \id - \A(\neg X) \right) \G(\neg X\times Y) =\nul
\end{equation*}
and similarly
\begin{equation*}
\G(\neg X\times Y) \A(X)=\nul.
\end{equation*}
It then follows that
\begin{eqnarray*}
\A(X)\B(Y) &=& \A(X)\G(\Omega_\A\times Y)=\A(X) \left( \G(X\times Y)+\G(\neg X\times Y) \right) \\
&=& \G(X\times Y)
\end{eqnarray*}
and similarly
\begin{equation*}
\B(Y)\A(X)=\G(X\times Y) \, .
\end{equation*}
A comparison of these equations shows that $\A$ and $\B$ commute and (\ref{eq:GAB}) holds. As $\G$ is determined by its values of product sets, it is unique.
\end{proof}

\section*{Acknowledgement}

This work was supported by projects RPEU-0014-06 and
QIAM APVV-0673-07.

\end{document}